\pdfoutput=1

\documentclass[envcntsame]{llncs}
\usepackage{amsfonts,amsmath}

\usepackage{graphicx}
\usepackage{subfig} 
\usepackage{url}
\usepackage{verbatim}
\usepackage{color}
\definecolor{lgray}{gray}{0.92}
\definecolor{lblue}{rgb}{0.90,0.90,1.00}
\definecolor{lyellow}{rgb}{1.00,1.00,0.70}

\usepackage{listings}
\newenvironment{codex}{\small\verbatim}{\endverbatim\normalsize}

\lstnewenvironment{code}
    {\lstset{}%
      \csname lst@SetFirstLabel\endcsname}
    {\csname lst@SaveFirstLabel\endcsname}
\newtheorem{prop}{Proposition}

\newtheorem{ex}{Example}
\newcommand{\BI}[0]{\begin{itemize}}
\newcommand{\EI}[0]{\end{itemize}}

\newcommand{\BE}[0]{\begin{enumerate}}
\newcommand{\EE}[0]{\end{enumerate}}

\newcommand{\BX}[0]{\begin{ex}}
\newcommand{\EX}[0]{\end{ex}}
\newcommand{\BV}[0]{\begin{verbatim}}
\newcommand{\EV}[0]{\end{verbatim}}

\def \bscale1 {0.25}
\def \bscale {0.25}
\def \N {\mathbb{N}}
\def \T {\tt T}
\def \M {\tt M}
\def \P {\tt P}

\def \C {\tt Cat}

\newcommand{\FIG}[4]{
\begin{figure}[htbp]
\centering
{\includegraphics[scale=#3]{figs/#4}}
\caption{#2}
\label{#1}
\end{figure}
}

\newcommand{\HFIGS}[7]{
\begin{figure}[htbp]
  \centering
  \subfloat[#3]{
    {\includegraphics[scale=#7]{figs/#5}}}   
  \subfloat[#4]{
    {\includegraphics[scale=#7]{figs/#6}}}
  \caption{#2}
  \label{#1}
\end{figure}
}

\begin{document}

\title{ 
  A Generic Numbering System based on Catalan Families of Combinatorial Objects
}

\author{Paul Tarau}
\institute{
   {Department of Computer Science and Engineering}\\
   {University of North Texas}\\ 
   {\em paul.tarau@unt.edu}\\
}

\maketitle

\date{}

\pagestyle{plain}

\begin{abstract}

We describe  arithmetic algorithms
on a canonical number representation based 
on the Catalan family of combinatorial objects
specified as a Haskell type class.

Our algorithms work on a {\em generic} representation
that we illustrate on instances members of the Catalan family,
like ordered binary and multiway trees. 
We validate the correctness of our
algorithms by defining an instance of the same type class
based  the usual bitstring-based natural numbers.
  
While their average and worst case
complexity is within
constant factors of their
traditional counterparts,
our algorithms provide
super-exponential gains   
for numbers corresponding to
Catalan objects of low
representation size.




\end{abstract}

\section{Introduction}

This paper generalizes the results of \cite{lata14} and \cite{ictac14tarau}, 
where special instances of the Catalan family \cite{stanleyEC,knuth_trees} 
of combinatorial objects 
(the language of balanced parentheses and the set ordered rooted trees
with empty leaves, respectively)  have been endowed 
with basic arithmetic operations corresponding to those on 
bitstring-represented natural numbers and extends \cite{padl16a} with
advanced arithmetic operations and computations involving record holding primes
and values of the Collatz sequence on gigantic numbers.

The main contribution of this paper is a {\em generic}
Catalan family based numbering
system that supports computations with numbers comparable in size
with Knuth's ``arrow-up'' notation.
These computations have an  average and worst case and complexity
that is comparable or better than the traditional  binary numbers,
while on neighborhoods of iterated powers of two they
 outperform binary numbers by an arbitrary
tower of exponents factor. 

As the Catalan family   contains a large
number of computationally isomorphic but structurally distinct
combinatorial objects, we will describe our
arithmetic computations generically, using Haskell's {\em type
classes} \cite{DBLP:conf/popl/WadlerB89}, of which typical members
of the Catalan family,
like binary trees and multiway trees will be described as instances.

At the same time, an {\em atypical instance} will be derived,
representing the set of {\em natural numbers} $\N$, which will be
used to cross-validate the correctness of our generically
defined arithmetic operations.

The paper is organized as follows.
Section \ref{related} discusses related work.
Section \ref{cats} introduces a generic view of Catalan families as
a Haskell type class, with subsection \ref{nats} embedding the set of natural
numbers as an instance of the family.
Section \ref{arith} introduces
basic algorithms for arithmetic operations taking advantage of our number
representation, with subsection \ref{succ} focusing on
constant time successor and predecessor operations.
Section \ref{addsub} 
describes   arithmetic operations
that favor operands of low representation complexity including
computations with giant numbers.
Section \ref{concl} concludes the paper. 

We have adopted
a {\em literate programming} style, i.e. the
the code described in the paper
forms a self-contained Haskell module (tested with ghc 7.10.2).
It is available 
at { \url{http://www.cse.unt.edu/~tarau/research/2014/GCat.hs}}.

\section{Related work} \label{related}

The first instance of a {\em hereditary number system} representing natural numbers
as multiway trees
occurs in the proof of Goodstein's theorem \cite{goodstein}, where
replacement of finite numbers on a tree's branches by the ordinal $\omega$
allows him to prove that a ``hailstone sequence'', after visiting arbitrarily
large numbers, eventually turns around and terminates.

Another hereditary number system is Knuth's TCALC program \cite{tcalc} that
decomposes $n=2^a+b$ with $0 \leq b<2^a$ and then recurses on $a$ and $b$ with the same decomposition. Given the constraint on $a$ and $b$, while hereditary, the TCALC system is not based on a bijection between $\N$ and $\N \times \N$ and therefore the representation is not a bijection. Moreover, the literate C-program that defines it
only implements successor, addition, comparison and multiplication, and does not
provide a constant time power of 2 and low complexity leftshift / rightshift 
operations.  

Several notations for very large numbers have been invented in the past. Examples
include Knuth's {\em up-arrow} notation \cite{knuthUp},
covering operations like the {\em tetration} (a notation for towers of exponents).
In contrast to the tree-based natural numbers we describe in this paper,
such notations are not closed under
addition and multiplication, and consequently
they cannot be used as a replacement
for ordinary binary or decimal numbers.

While combinatorial enumeration and combinatorial generation, for which a
vast literature exists (see for 
instance  \cite{stanleyEC}, \cite{combi99} and \cite{knuth_trees})
can be seen as providing unary Peano arithmetic operations implicitly,
like in \cite{lata14} and \cite{ictac14tarau}, the algorithms in this paper
enable arithmetic
computations of efficiency comparable to the usual
binary numbers (or better) using combinatorial families.

Providing a {\em generic mechanism} for efficient arithmetic computations
with {\em arbitrary members of the Catalan family} is the main motivation and 
the most significant contribution of this paper. It is notationally similar
to the type class mechanism sketched in \cite{bintrees}. Unfortunately,
the simpler binary tree-based computation of \cite{bintrees} does not support the
$O(log^*(n))$ successor and predecessor operations described in this paper,
that facilitate size proportionate encodings of data types.
Such encodings are the critical component of the companion paper 
\cite{padl16b}, which describes their application  to lambda terms.

\section{The Catalan family of combinatorial objects} \label{cats}

The Haskell data type {\tt T} representing ordered rooted binary trees
with empty leaves {\tt E} and branches provided by the constructor {\tt C}
is a typical member of the Catalan family of combinatorial objects \cite{stanleyEC}.
\begin{code}
data T = E | C T T deriving (Eq,Show,Read)
\end{code}
Note the use of the type classes {\tt Eq, Show} and {\tt Read}
to derive structural equality and respectively human readable
output and input for this data type.

The data type {\tt M} is another well-known member of the Catalan family, 
defining multiway ordered rooted trees with empty leaves.
\begin{code}  
data M = F [M] deriving (Eq,Show,Read)
\end{code}

\subsection{A generic view of Catalan families as a Haskell type class } \label{class}
We will work through the paper with a generic data type ranging over instances
of the type class {\tt Cat}, representing a member of the Catalan family
of combinatorial objects \cite{stanleyEC}.
\begin{code}
class (Show a,Read a,Eq a) => Cat a where
  e :: a
  
  c :: (a,a) -> a
  c' :: a -> (a,a) 
\end{code}
The zero element is denoted {\tt e} and
we inherit from classes {\tt Read} and {\tt Show} which
ensure derivation of input and output functions
for members of type class {\tt Cat}
as well as from type class {\tt Eq} that ensures derivation of the
structural equality predicate {\verb~==~} and its negation \verb~/=~.

We will also define the corresponding recognizer predicates
{\tt e\_} and {\tt c\_}, relying on the derived equality
relation inherited from the Haskell type class {\tt Eq}.
\begin{code}    
  e_ :: a -> Bool  
  e_ a = a == e
  
  c_ :: a -> Bool
  c_ a = a /= e
\end{code}

For each instance, we assume that {\tt c} and {\tt c'} are inverses on their respective domains {\tt Cat} $\times$ {\tt Cat} and {\tt Cat - \{e\}}, and {\tt e} is distinct from objects constructed with {\tt c}, more precisely that the following hold:

\begin{equation} \label{inv}
\forall x.~c' (c~x)=x  \land \forall y.~(c\_ ~y \Rightarrow c~(c'~y) = y)
\end{equation}

\begin{equation} \label{excl}
\forall x.~ (e\_~x \lor c\_~x) \land \lnot ( e\_~x \land c\_~x )
\end{equation}

When talking about ``objects of type {\tt Cat}'' we will actually mean 
an instance {\tt a} of the polymorphic type {\tt Cat a} that verifies 
equations (\ref{inv}) and (\ref{excl}).

\subsection{The instance $\T$ of ordered rooted binary trees}

The operations defined
in type class {\tt Cat}
correspond naturally to the ordered rooted binary tree view of the Catalan
family, materialized as the data type $\T$.
\begin{code} 
instance Cat T where
  e = E
  
  c (x,y) = C x y 
  
  c' (C x y) = (x,y)
\end{code}
Note that adding and removing the constructor  {\tt C} trivially verifies
the assumption that our generic operations {\tt c} and {\tt c'}
are inverses\footnote{In fact, one can see the functions {\tt e, e\_ , c, c', c\_}
as a generic API 
abstracting away the essential properties
of the constructors {\tt E} and {\tt C}.
}.

\subsection{The instance $\M$ of ordered rooted multiway trees}
The alternative view of the Catalan family as multiway trees is
materialized as the data type $\M$.
\begin{code}
instance Cat M where
  e = F []
  c (x,F xs) = F (x:xs)
  c' (F (x:xs)) = (x,F xs)
\end{code}
Note that the
assumption that our generic operations {\tt c} and {\tt c'}
are inverses is easily verified in this case as well, given
the bijection between binary and multiway trees. 
Moreover, note that operations
on types {\tt T} and {\tt M}, expressed
in terms of their generic
type class {\tt Cat} counterparts, result
in a constant extra effort. Therefore, we will
safely ignore it when discussing the
complexity of different operations.

\subsection{An unusual member of the Catalan family: the set of natural numbers $\N$} \label{nats}

The  (big-endian) binary representation of a natural number can be written as a concatenation of binary digits of the form
\begin{equation} \label{bin}
n=b_0^{k_0}b_1^{k_1}\ldots b_i^{k_i} \ldots b_m^{k_m}
\end{equation}
with  $b_i \in \{0,1\}$,$~b_i \neq b_{i+1}$ and the highest digit  $b_m=1$.
The following hold.
\begin{prop}
An even number of the form $0^ij$ corresponds
to the operation $2^ij$ and an odd number of the form $1^ij$ corresponds
to the operation $2^i(j+1)-1$. 
\end{prop}
\begin{proof}
It is clearly the case that $0^ij$ corresponds to multiplication by a power of $2$. If $f(i)=2i+1$ then it is shown by induction (see \cite{sac14}) that the $i$-th iterate of $f$, $f^i$ is computed as in the equation (\ref{fiter})
\begin{equation}
f^i(j)=2^i(j+1)-1 \label{fiter}
\end{equation}
Observe that each block $1^i$ in $n$, represented as $1^ij$ in equation (\ref{bin}), corresponds to the iterated application of $f$, $i$ times,
$n=f^i(j)$.
\end{proof}

\begin{prop} \label{parity}
A number $n$ is even if and only if it contains an even number of blocks of the form 
$b_i^{k_i}$ in equation (\ref{bin}). 
A number $n$ is odd if and only if it contains an odd number of blocks of the form 
$b_i^{k_i}$ in equation (\ref{bin}). 
\end{prop}
\begin{proof}
It follows from the fact that the highest digit (and therefore the last block in big-endian representation) is $1$ and the parity of the blocks alternate.
\end{proof}

This suggests defining the {\tt c} operation of type class {\tt Cat}
as follows.
\begin{equation}\label{cnat}
c(i,j)=
\begin{cases}
{2^{i+1}j}  & {\text{if } j \text{ is odd}},\\
{2^{i+1}(j+1)-1} & {\text{if } j \text{ is even}}.
\end{cases}
\end{equation}
Note that the exponents are $i+1$ instead of $i$ as we start counting at $0$. Note also
that  $c (i,j)$ will be even when $j$ is odd and odd when $j$ is even.

\begin{prop}
The equation (\ref{cnat}) defines a bijection  $c: \N \times \N \to  \N^+=\N-\{0\}$.
\end{prop}
Therefore {\tt c} has an inverse {\tt c'}, that we will constructively define
together with {\tt c}. 
The following Haskell code
defines the instance of the Catalan family corresponding to $\N$.
\begin{code}
type N = Integer
instance Cat Integer where
  e = 0

  c (i,j) | i>=0 && j>=0 = 2^(i+1)*(j+b)-b where b = mod (j+1) 2
\end{code}
The definition of the inverse {\tt c'} relies on the {\em dyadic valuation} of a number $n$, $\nu_2(n)$, defined as the largest exponent of 2 dividing $n$, implemented as the helper function {\tt dyadicVal}. 
\begin{code}   
  c' k | k>0 = (i-1,j-b) where
    b = mod k 2
    (i,j) = dyadicVal (k+b)

    dyadicVal k | even k = (1+i,j) where  (i,j) = dyadicVal (div k 2)
    dyadicVal k = (0,k)  
\end{code}
Note the use of the parity {\tt b} in both definitions, which differentiates between
the computations for {\em even} and {\em odd} numbers.

The following examples illustrate the use of {\tt c} and {\tt c'} on this instance.
\begin{codex}
*GCat> c (100,200)
509595541291748219401674688561151
*GCat> c' it
(100,200)
*GCat> map c' [1..10]
[(0,0),(0,1),(1,0),(1,1),(0,2),(0,3),(2,0),(2,1),(0,4),(0,5)]
*GCat> map c it
[1,2,3,4,5,6,7,8,9,10]
\end{codex}

Figure \ref{cats2018} illustrates the DAG obtained by applying the operation
{\tt c'} repeatedly and merging identical subtrees for three consecutive numbers. 
The order of the edges is
marked with {\tt 0} and {\tt 1}.
\FIG{cats2018}{DAG representing the list [2018..2022]}{0.30}{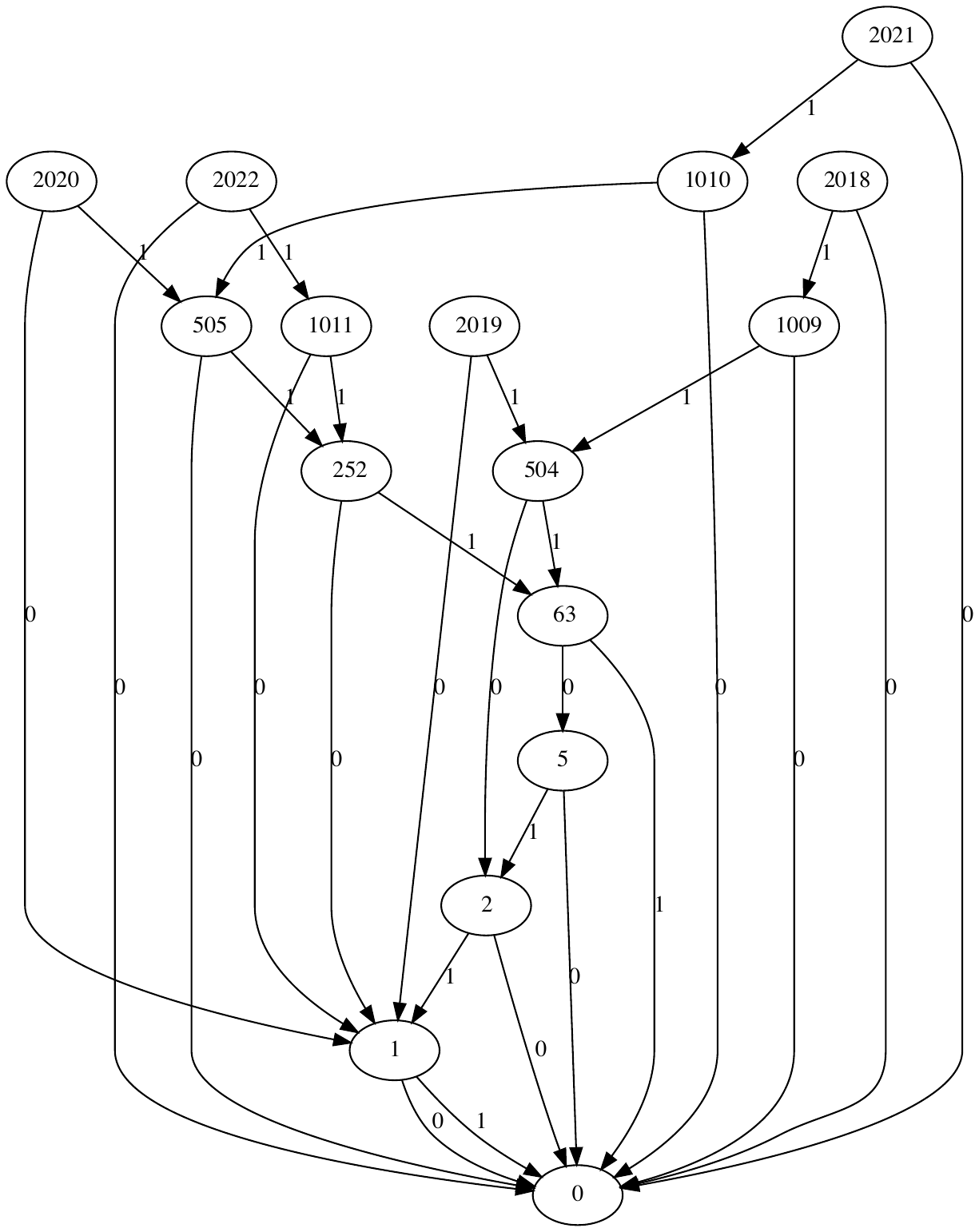}

\subsection{Other instances using members of the Catalan family of combinatorial objects}

The language of balanced parentheses (Dyck words) is another member of the Catalan family, easily seen as in bijection with multiway trees. A convenient representation is one in which
a list of positive integers are used to mark the position of each enclosing parenthesis, marked each with a {\tt 0}. As the offset between opening and closing parenthesis is always even, we adjust that with a linear transformation ensuring the range is the positive integers. Thus \verb~((()())(()))~ (corresponding to {\tt 24}) is represented as {\tt [3,1,0,1,0,0,2,1,0,0]} and \verb~(()()()()()())~, corresponding to {\tt 42} is represented as {\tt [1,0,1,0,1,0,1,0,1,0,1,0]}.

After defining a data type wrapping list of integers, we obtain the instance {\tt P} of type class {\tt Cat}.

\begin{code}
data P = P [Int] deriving (Eq,Show,Read)

instance Cat P where
  e = P []
  c (P xs,P ys) = P ((l:xs)++(0:ys)) where l=1+((length xs) `div` 2)
  c' (P (l:xs)) = (P (take (x-1) xs),P (drop x xs)) where x = 2*l-1
\end{code}

\subsection{The transformers: morphing between instances of the Catalan family}

As all our instances implement the bijection {\tt c} and its inverse {\tt c'}, a generic transformer from an instance to another is defined
by the function {\tt view}:
\begin{code}
view :: (Cat a, Cat b) => a -> b
view z | e_ z = e
view z | c_ z = c (view x,view y) where (x,y) = c' z
\end{code}
To obtain transformers defining bijections with $\N, \T$ and $\M$ as ranges, we
will simply provide specialized type declarations for them:
\begin{code}
n :: Cat a => a->N
n = view
\end{code}
\begin{code}
t :: Cat a => a->T
t = view
\end{code}
\begin{code}
m :: Cat a => a->M
m = view
\end{code}
\begin{code}  
p :: Cat a => a->P
p = view 
\end{code}
The following examples illustrate the resulting specialized conversion
functions:
\begin{codex}
*GCat> t 42
C E (C E (C E (C E (C E (C E E)))))
*GCat> m it
F [F [],F [],F [],F [],F [],F []]
*GCat> p it
P [1,0,1,0,1,0,1,0,1,0,1,0]
*GCat> n it
42
\end{codex}

A list view of an instance of type class {\tt Cat} is obtained by
iterating the constructor {\tt c} and its inverse {\tt c'}.
\begin{code}
to_list :: Cat a => a -> [a]
to_list x | e_ x = []
to_list x | c_ x  = h:hs where 
    (h,t) = c' x
    hs = to_list t
\end{code}
\begin{code}
from_list :: Cat a => [a] -> a
from_list [] = e
from_list (x:xs) = c (x,from_list xs)
\end{code}
They work as follows:
\begin{codex}
*GCat> to_list 2020
[1,0,1,5]
*GCat> from_list it
2020
\end{codex}
The function {\tt to\_list} corresponds to the children of a node in the multiway tree view provided by instance {\tt M}. Along the lines of
 \cite{calc09fiso,ppdp14tarau} one can use
{\tt to\_list} and {\tt from\_list}
to define size-proportionate bijective encodings of sets, multisets and data types built 
from them.

The function {\tt catShow} provides a view as a string of balanced parentheses, an instance of the Catalan family for which arithmetic computations 
are introduced in \cite{lata14}.
\begin{code}   
catShow :: Cat a => a -> [Char]
catShow x | e_ x = "()"
catShow x | c_ x = r where
    xs = to_list x
    r = "(" ++ (concatMap catShow xs) ++ ")"
\end{code}
It is illustrated below.
\begin{codex}
*GCat> catShow 0
"()"
*GCat> catShow 1
"(())"
*GCat> catShow 12345
"(()(())(()())(()()())(()))"
\end{codex}

\section{Generic arithmetic operations on members of the Catalan family} \label{arith}

We will now implement arithmetic operations on Catalan families,
generically, in terms of the operations on type class {\tt Cat}.

\subsection{Basic utilities}
We start with some simple functions to be used later.

\subsubsection{Inferring even and odd}
As we know for sure that the instance $\N$, corresponding to
natural numbers supports arithmetic operations, we will
mimic their behavior  at the level of
the type class {\tt Cat}.

The operations {\tt even\_} and {\tt odd\_} implement the
observation following from of Prop. \ref{parity} that parity (staring with $1$
at the highest block) alternates with each block of distinct $0$ or $1$ digits.
\begin{code}
even_ :: Cat a => a -> Bool
even_ x | e_ x = True
even_ z | c_ z = odd_ y where (_,y)=c' z

odd_ :: Cat a => a -> Bool
odd_ x | e_ x = False
odd_ z | c_ z = even_ y where (_,y)=c' z
\end{code}

\subsubsection{One}
We also provide a constant {\tt u} and
a recognizer predicate {\tt u\_} for $1$.
\begin{code}
u :: Cat a => a
u = c (e,e)

u_ :: Cat a => a-> Bool
u_ z = c_ z && e_ x && e_ y where (x,y) = c' z
\end{code}

\subsection{Average constant time successor and predecessor} \label{succ}
We will now specify successor and predecessor on the family of data types $\C$
through two mutually recursive functions, {\tt s} and {\tt s'}.
They are based on arithmetic observations
about the behavior of these blocks when incrementing or decrementing
a binary number by {\tt 1}, derived from equation (\ref{cnat}).

They first decompose their arguments using {\tt c'}. Then,
after transforming them as a result of adding or subtracting {\tt 1},
they place back the results with the {\tt c} operation.

Note that the two functions work {\em on a block of {\tt 0} or {\tt 1}
digits at a time}. 
The main intuition is that as adding 
or subtracting {\tt 1} changes the parity of a number
and as carry-ons propagate over a block of {\tt 1}s in the case of addition
and over a block of {\tt 0}s in the case of subtraction,
{\em blocks} of contiguous {\tt 0} and {\tt 1} 
digits will be flipped as a result of
applying {\tt s} or {\tt s'}.

\begin{code}
s :: Cat a => a -> a 
s x | e_ x = u -- 1
s z | c_ z && e_ y = c (x,u) where -- 2
   (x,y) = c' z
\end{code}
For the general case, the successor function {\tt s} delegates the transformation of
the blocks of $0$ and $1$ digits to functions {\tt f} and {\tt g}
handling {\tt even\_} and respectively {\tt odd\_} cases. 
\begin{code}
s a | c_ a = if even_ a then f a else g a where

   f k | c_ w && e_ v = c (s x,y) where -- 3
    (v,w) = c' k
    (x,y) = c' w
   f k = c (e, c (s' x,y)) where -- 4
     (x,y) = c' k     
     
   g k | c_ w && c_ n && e_ m = c (x, c (s y,z)) where -- 5
    (x,w) = c' k
    (m,n) = c' w
    (y,z) = c' n  
   g k | c_ v = c (x, c (e, c (s' y, z))) where -- 6
    (x,v) = c' k
    (y,z) = c' v
\end{code}
The predecessor function {\tt s'} inverts the work of {\tt s}
as marked by a comment of the form \verb~k --~, for \verb~k~
ranging from {\tt 1} to {\tt 6}.
\begin{code}    
s' :: Cat a => a -> a
s' k | u_ k = e where -- 1
    (x,y) = c' k  
s' k | c_ k && u_ v = c (x,e) where -- 2
    (x,v) = c' k 
\end{code}
For the general case,  {\tt s'} delegates the transformation of
the blocks of $0$ and $1$ digits to functions {\tt g} and {\tt f}
handling {\tt even\_} and respectively {\tt odd\_} cases.
\begin{code}
s' a | c_ a = if even_ a then g' a else f' a where

     g' k | c_ v && c_ w && e_ r = c (x, c (s y,z)) where -- 6
       (x,v) = c' k
       (r,w) = c' v
       (y,z) = c' w    
     g' k  | c_ v = c (x,c (e, c (s' y, z))) where -- 5
       (x,v) = c' k
       (y,z) = c' v     
       
     f' k | c_ v && e_ r = c (s x,z) where -- 4
        (r,v) = c' k
        (x,z) = c' v
     f' k =  c (e, c (s' x,y)) where -- 3
        (x,y) = c' k
\end{code}
One can see that their use matches successor and predecessor on instance {\tt N}:
\begin{codex}
*GCat> map s [0..15]
[1,2,3,4,5,6,7,8,9,10,11,12,13,14,15,16]
*GCat> map s' it
[0,1,2,3,4,5,6,7,8,9,10,11,12,13,14,15]
\end{codex}

\begin{prop}
Denote $\C^+=\C-\{\text{\tt e}\}$.
The functions $s:\C \to \C^+$ and $s':\C^+ \to \C$ are inverses.
\end{prop}
\begin{proof}
For each instance of $\C$, it follows by structural induction after observing that 
patterns for rules marked with the number {\tt -- k}
in {\tt s} correspond one by one to patterns marked by
{\tt -- k} in {\tt s'} and vice versa.
\end{proof}

More generally, it can be shown that Peano's axioms hold
and as a result $<\C,e,s>$ is a {\em Peano algebra}.
This is expected, as {\tt s} provides a combinatorial enumeration
of the infinite stream of Catalan objects, as illustrated below
on instance {\tt T}:
\begin{codex}
Cats> s E
C E E
*GCat> s it
C E (C E E)
*GCat> s it
C (C E E) E
\end{codex}

The function {\tt nums} generates an initial segment of
the ``natural numbers'' defined by an instance of {\tt Cat}.
\begin{code}
nums :: Cat a => a -> [a]
nums x = f x [] where 
  f x xs | e_ x = e:xs
  f x xs = f (s' x) (x:xs)
\end{code}

{\em Note that if parity information is kept explicitly, the calls to {\tt odd\_} and {\tt even\_} are constant time}, as we will assume in the rest of the paper.
We will also assume, that when complexity is discussed, 
{\em a representation like the tree data types
$\T$ or $\M$ are used, making the operations {\tt c} and {\tt c'} constant time}.
Note also that this is clearly not the case for the instance {\tt N} using
the traditional bitstring representation
where effort proportional to the length of the bitstring may be involved.

\begin{prop}
The worst case time complexity of the {\tt s} and {\tt s'} operations
on an input $n$ is given by the {\em iterated logarithm} $O(log_2^*(n))$.
\end{prop}
\begin{proof}
Note that calls to {\tt s,s'} in {\tt s} or {\tt s'} happen
on terms
at most logarithmic in the bitsize of their operands. 
The recurrence relation
counting the worst case number of calls to {\tt s} or {\tt s'} is:
$T(n)=T({{log_2(n)}})+O(1)$, which solves to $T(n)=O(log_2^*(n))$.
\end{proof}
Note that this is much better than the logarithmic worst case for binary numbers (when computing, for instance,
binary 111...111+1=1000...000).

\begin{prop} \label{avg}
{\tt s} and {\tt s'} are constant time, on the average. 
\end{prop}

\begin{proof}
When computing
the successor or predecessor of a number of bitsize $n$,
calls to {\tt s,s'} in {\tt s} or {\tt s'} happen on at most one subterm.
Observe that the average size of a contiguous block of {\tt 0}s or {\tt 1s} in a
number of bitsize $n$ has the upper bound 2 as
${\sum_{k=0}^{n} {1 \over 2^k}} = 2-{1 \over {2^n}} < 2$.
As on 2-bit numbers we have an average of $0.25$ more calls, we can conclude that the total average number of calls is constant, with upper bound $2+0.25=2.25$.
\end{proof}

A quick empirical evaluation confirms this.
When computing the successor on 
the first $2^{30}=1073741824$ natural 
numbers, there are in total
2381889348 calls to {\tt s} and {\tt s'}, 
averaging to 2.2183 per computation.
The same average for $100$ successor computations on
$5000$ bit random numbers oscillates around $2.22$.

\subsection{A few other average constant time, $O(log^*)$ worst case operations}
We will derive a few operations that inherit their complexity
from {\tt s} and {\tt s'}.

\subsubsection{Double and half}

Doubling a number {\tt db}
and reversing the {\tt db} operation ({\tt hf}) are
quite simple. For instance, {\tt db} proceeds
by adding a new counter for odd numbers
and incrementing the first counter for even ones. 

\begin{code}
db :: Cat a => a -> a
db x | e_ x  = e
db x | odd_ x = c (e,x)
db z = c (s x,y) where (x,y) = c' z
\end{code}

\begin{code}
hf :: Cat a => a -> a
hf x | e_ x = e
hf z | e_ x = y where (x,y) = c' z
hf z  = c (s' x,y) where (x,y) = c' z
\end{code}

\subsubsection{Power of 2 and its left inverse}
Note that such efficient implementations follow directly from
simple number theoretic observations. 

For instance, 
{\tt exp2}, computing an power of $2$ ,
has the following definition in terms of {\tt c} and {\tt s'} from which
it inherits its complexity up to a constant factor.
\begin{code}
exp2 :: Cat a => a -> a
exp2 x | e_ x = u
exp2 x = c (s' x, u)
\end{code}
The same applies to its left inverse {\tt log2}:
\begin{code}
log2 :: Cat a => a -> a
log2 x | u_ x = e
log2 x | u_ z = s y where (y,z) = c' x
\end{code}
\begin{prop}
The operations {\tt db}, {\tt hf}, {\tt exp2} and {\tt log2} are average constant time and are ${log}^*$ in the worst case.
\end{prop}
\begin{proof}
At most one call to {\tt s,s'} is made in each definition. Therefore these operations have the same worst and average complexity as {\tt s} and {\tt s'}.
\end{proof}
We illustrate their work on instances {\tt N}:
\begin{codex}
*GCat> map exp2 [0..14]
[1,2,4,8,16,32,64,128,256,512,1024,2048,4096,8192,16384]
*GCat> map log2 it
[0,1,2,3,4,5,6,7,8,9,10,11,12,13,14]
\end{codex}
More interestingly, a tall tower of exponents
that would overflow memory on instance {\tt N}, is easily supported
on instances {\tt T} and {\tt M} as shown below:
\begin{codex}
*GCat> exp2 (exp2 (exp2 (exp2 (exp2 (exp2 (exp2 E))))))
C (C (C (C (C (C E E) E) E) E) E) (C E E)
*GCat> m it
F [F [F [F [F [F [F []]]]]],F []]
*GCat> log2 (log2 (log2 (log2 (log2 (log2 (log2 it))))))
F []
\end{codex}
{\em
This example illustrates the main motivation for defining
arithmetic computation with the ``typical''
members of the Catalan family: their 
ability to deal with giant numbers.
}

Another average constant time / worst case $log^*$ algorithm is counting the trailing {\tt 0}s of a number (on instance $\T$):
\begin{code}
trailingZeros x | e_ x = e
trailingZeros x | odd_ x = e
trailingZeros x = s (fst (c' x))
\end{code}
This contrasts with the O(log(n) worst case performance to 
count them with a bitstring representation.

\section{Addition, subtraction and their mutually recursive helpers} \label{addsub}

We will derive in this section efficient addition and subtraction
that {\em work on one
run-length compressed block at a time}, rather than by
individual {\tt 0} and {\tt 1} digit steps.

\subsection{Multiplication by a power of {\tt 2}}

We start with the functions {\tt leftshiftBy}, {\tt leftshiftBy'} 
and  {\tt leftshiftBy}'' corresponding respectively
to $(\lambda x.2x)^n(k)$ , $(\lambda x.2x+1)^n(k)$ and $(\lambda x.2x+2)^n(k)$.

The function {\tt leftshiftBy} prefixes an odd number with
a block of $1$s and extends a block of $0$s by incrementing
their count.
\begin{code}
leftshiftBy :: Cat a => a -> a -> a
leftshiftBy x y | e_ x = y
leftshiftBy _ y | e_ y = e
leftshiftBy x y | odd_ y = c (s' x, y) 
leftshiftBy x v = c (add x y, z) where (y,z) = c' v
\end{code}
The function {\tt leftshiftBy'} is based on
equation (\ref{onk}).

\begin{equation}\label{onk}
(\lambda x.2x+1)^n(k)=2^n(k+1)-1
\end{equation}

\begin{code}
leftshiftBy' :: Cat a => a -> a -> a
leftshiftBy' x k = s' (leftshiftBy x (s k)) 
\end{code}
The function {\tt leftshiftBy'} is based on
equation (\ref{ink}) (see \cite{sac14} for a direct proof by induction).
\begin{equation}\label{ink}
(\lambda x.2x+2)^n(k)=2^n(k+2)-2
\end{equation}
\begin{code}
leftshiftBy'' :: Cat a => a -> a -> a
leftshiftBy'' x k = s' (s' (leftshiftBy x (s (s k))))
\end{code}

They are part of a {\em chain of mutually recursive functions} as
they are already referring to the
{\tt add} function.
Note also that instead of naively iterating, they implement a
more efficient algorithm, working
 ``one block at a time''. For instance, when detecting that its argument counts
a number of {\tt 1}s, {\tt leftshiftBy'} just increments that count.
As a result, the algorithm favors numbers with relatively few large blocks of
{\tt 0} and {\tt 1} digits. 

While not directly used in the addition operation it is interesting to
observe that division by a power of 2 can also be computed efficiently.

\subsection{An inverse operation: division by a power of 2}
The function {\tt rightshiftBy} goes over its argument {\tt y}
one block at a time, by comparing the size of the block and its argument
{\tt x} that is decremented after each block by the size of the block.
The local function {\tt f} handles the details, irrespectively of
the nature of the block, and stops when the
argument is exhausted.
More precisely, based on the result {\tt EQ, LT, GT} of the comparison,
{\tt f} either stops or, calls {\tt rightshiftBy} on the
the value of {\tt x} reduced by the size of the block {\tt a' = s a}.

\begin{code}
rightshiftBy :: Cat a => a -> a -> a
rightshiftBy x y | e_ x = y
rightshiftBy _ y | e_ y = e
rightshiftBy x y = f (cmp x a')  where
  (a,b) = c' y
  a' = s a
  f LT = c (sub a x,b) 
  f EQ = b
  f GT = rightshiftBy (sub  x a') b
\end{code}

\subsection{Addition optimized for numbers built from a few large blocks of 0s and 1s}
We are now ready to define addition. The base cases are
\begin{code}
add :: Cat a => a -> a -> a
add x y | e_ x = y
add x y | e_ y = x
\end{code}
In the case when both terms represent even numbers, the two blocks add up
to an even block of the same size. Note the use of {\tt cmp} 
and {\tt sub} in helper function {\tt f} to
trim off the larger block such that we can operate 
on two  blocks of equal size.
\begin{code}
add x y |even_ x && even_ y = f (cmp a b) where
  (a,as) = c' x
  (b,bs) = c' y
  f EQ = leftshiftBy (s a) (add as bs)
  f GT = leftshiftBy (s b) (add (leftshiftBy (sub a b) as) bs)
  f LT = leftshiftBy (s a) (add as (leftshiftBy (sub b a) bs))
\end{code}
In the case when the first term is even and the second odd, the two blocks add up
to an odd block of the same size.  
\begin{code}
add x y |even_ x && odd_ y = f (cmp a b) where
  (a,as) = c' x
  (b,bs) = c' y
  f EQ = leftshiftBy' (s a) (add as bs)
  f GT = leftshiftBy' (s b) (add (leftshiftBy (sub a b) as) bs)
  f LT = leftshiftBy' (s a) (add as (leftshiftBy' (sub b a) bs))
\end{code}
In the case when the second term is even and the first odd the two blocks 
also add up
to an odd block of the same size.
\begin{code}
add x y |odd_ x && even_ y = add y x
\end{code}
In the case when both terms represent odd numbers,
we use the  identity (\ref{oplus}):
\begin{equation} \label{oplus}
{(\lambda x.2x+1)}^k(x) + {(\lambda x.2x+1)}^k(y) = {(\lambda x.2x+2)}^k(x+y) 
\end{equation}
\begin{code}
add x y | odd_ x && odd_ y = f (cmp a b) where
  (a,as) = c' x
  (b,bs) = c' y
  f EQ =  leftshiftBy'' (s a) (add as bs)
  f GT =  leftshiftBy'' (s b)  (add (leftshiftBy' (sub a b) as) bs)
  f LT =  leftshiftBy'' (s a)  (add as (leftshiftBy' (sub b a) bs))
\end{code}
Note the presence of the comparison operation {\tt cmp} also
part of our chain of mutually recursive operations.
Note also the local function {\tt f} that in each case ensures that a block of the same
size is extracted, depending on which of the two operands {\tt a} or {\tt b} is larger.

Details of implementation for subtraction 
are quite similar to those of the addition. The curious reader
can find the implementation and the literate programming explanations of the
complete set of arithmetic operations with the type class {\tt Cat} in
our arxiv draft at \cite{arxiv_cats}.

\subsection*{Comparison}
The comparison operation {\tt cmp}
provides a total order (isomorphic to that on $\N$) on our generic type $\C$.
It relies on {\tt bitsize} computing the number of binary digits 
constructing a term in $\C$,
also part of our mutually recursive functions, to be defined later.

We first observe that only terms of the same bitsize need detailed comparison,
otherwise the relation between their bitsizes is enough, {\em recursively}.
More precisely, the following holds:
\begin{prop} \label{bitineq}
Let {\tt bitsize} count the number of digits of a base-2 number,
with the convention that it is {\tt 0} for {\tt 0}.
Then {\tt bitsize}$(x) <${\tt bitsize}$(y) \Rightarrow x<y$.
\end{prop}
\begin{proof}
Observe that their lexicographic enumeration ensures that the
bitsize of base-2 numbers is a non-decreasing function.
\end{proof}

The comparison operation also proceeds one block at a time, and it
also takes some inferential shortcuts, when possible.
\begin{code}
cmp :: Cat a=> a->a->Ordering
cmp x y | e_ x && e_ y = EQ
cmp x _ | e_ x = LT
cmp _ y | e_ y = GT
cmp x y | u_ x && u_ (s' y) = LT
cmp  x y | u_ y && u_ (s' x) = GT
\end{code}
For instance, it is easy to see that comparison of {\tt x} and {\tt y} can be
reduced to comparison of bitsizes when they are distinct. 
Note that {\tt bitsize},
to be defined later,
is part of the chain of our mutually recursive functions.
\begin{code}
cmp x y | x' /= y'  = cmp x' y' where
  x' = bitsize x
  y' = bitsize y
\end{code}
When bitsizes are equal, a more elaborate
comparison needs to be done, delegated to function {\tt compBigFirst}.
\begin{code}  
cmp xs ys =  compBigFirst True True (rev xs) (rev ys) where
  rev = from_list . reverse . to_list
\end{code}
The function {\tt compBigFirst} compares two terms known to have the same
{\tt bitsize}. It works on reversed (highest order digit first) variants,
computed by {\tt reverse} and it takes  advantage of the
block structure using the following proposition:
\begin{prop}
Assuming two terms of the same bitsizes, the one with $1$ as its first before the
highest order digit, is larger than the one with $0$ as its first before the highest order digit.
\end{prop}
\begin{proof}
Observe that little-endian 
numbers obtained by applying the function {\tt rev}
are lexicographically ordered with $0<1$.
\end{proof}

As a consequence, {\tt cmp} only recurses when {\em identical} blocks
lead the sequence of blocks, otherwise it infers the {\tt LT} or {\tt GT}
relation. 
\begin{code}
  compBigFirst _ _ x y | e_ x && e_ y = EQ
  compBigFirst False False x y = f (cmp a b) where
    (a,as) = c' x
    (b,bs) = c' y
    f EQ = compBigFirst True True as bs
    f LT = GT
    f GT = LT   
  compBigFirst True True x y = f (cmp a b) where
    (a,as) = c' x
    (b,bs) = c' y
    f EQ = compBigFirst False False as bs
    f LT = LT
    f GT = GT
  compBigFirst False True x y = LT
  compBigFirst True False x y = GT
\end{code}
The following examples illustrate the agreement of {\tt cmp} with
the usual order relation on $\N$.
\begin{codex}
*GCat> cmp 5 10
LT
*GCat> cmp 10 10
EQ
*GCat> cmp 10 5
GT
\end{codex}

Note that the complexity of the comparison operation is proportional to the size
of the (smaller) of the two trees.

The function {\tt bitsize}, last in our chain of mutually recursive functions, 
 computes  the number of digits, except that we
define it as {\tt e} for constant function {\tt e} (corresponding to {\tt 0}). 
It works by summing up the counts of {\tt 0} and {\tt 1} digit
blocks composing a tree-represented natural number.
\begin{code}
bitsize :: Cat a => a -> a
bitsize z | e_ z = z
bitsize  z = s (add x (bitsize y)) where (x,y) = c' z
\end{code}

It follows that the  base-2 integer logarithm is  computed as
\begin{code}
ilog2 :: Cat a => a->a 
ilog2 = s' . bitsize
\end{code}

\subsubsection{Subtraction}
The code for the subtraction function {\tt sub} is similar to that for addition:
\begin{code}
sub :: Cat a => a -> a -> a
sub x y | e_ y = x
sub x y | even_ x && even_ y = f (cmp a b) where
  (a,as) = c' x
  (b,bs) = c' y
  f EQ = leftshiftBy (s a) (sub as bs)
  f GT = leftshiftBy (s b) (sub (leftshiftBy (sub a b) as) bs)
  f LT = leftshiftBy (s a) (sub as (leftshiftBy (sub b a) bs))
\end{code}
The case when both terms represent {\tt 1} blocks the result
is a {\tt 0} block:
\begin{code}
sub x y | odd_ x && odd_ y = f (cmp a b) where
  (a,as) = c' x
  (b,bs) = c' y
  f EQ = leftshiftBy (s a) (sub  as bs)
  f GT = leftshiftBy (s b) (sub (leftshiftBy' (sub a b) as) bs)
  f LT = leftshiftBy (s a) (sub as (leftshiftBy' (sub b a) bs))
\end{code}
The case when the first block is {\tt 1} and the second is
a {\tt 0} block is a {\tt 1} block:
\begin{code}
sub x y | odd_ x && even_ y = f (cmp a b) where
  (a,as) = c' x
  (b,bs) = c' y
  f EQ = leftshiftBy' (s a) (sub as bs)
  f GT = leftshiftBy' (s b) (sub (leftshiftBy' (sub a b) as) bs)
  f LT = leftshiftBy' (s a) (sub as (leftshiftBy (sub b a) bs)) 
\end{code}
Finally, when the first block is {\tt 0} and the second is {\tt 1}
an identity dual to (\ref{oplus}) is used:
\begin{code}
sub x y | even_ x && odd_ y = f (cmp a b) where
  (a,as) = c' x
  (b,bs) = c' y  
  f EQ = s (leftshiftBy (s a) (sub1 as bs))
  f GT = s (leftshiftBy (s b) (sub1 (leftshiftBy (sub a b) as) bs))
  f LT = s (leftshiftBy (s a) (sub1 as (leftshiftBy' (sub b a) bs)))

  sub1 x y = s' (sub x y)  
\end{code}
Note that these algorithms collapse to the ordinary binary addition and subtraction most
of the time, given that the average size of a block of contiguous {\tt 0}s or {\tt 1s} 
is
{\tt 2} bits (as shown in Prop. \ref{avg}), 
so their average complexity is within constant factor of their ordinary counterparts.

On the other hand,
as they are limited by the representation size of the operands
rather than their bitsize,
when compared with their bitstring counterparts,
these algorithms
 favor deeper trees made of large blocks, representing
 giant ``towers of exponents''-like numbers by working (recursively) one block at
 a time rather than 1 bit at a time, resulting in
 possibly super-exponential gains on them.

The following examples illustrate the agreement with their usual counterparts:
\begin{codex}
*GCat> map (add 10) [0..15]
[10,11,12,13,14,15,16,17,18,19,20,21,22,23,24,25]
*GCat> map (sub 15) [0..15]
[15,14,13,12,11,10,9,8,7,6,5,4,3,2,1,0]
\end{codex}

\section{Algorithms for advanced arithmetic operations} \label{muldiv}

\subsection{Multiplication, optimized for large blocks of {\tt 0}s and {\tt 1}s} \label{mul}

Devising a similar optimization as for {\tt add} and {\tt sub} for multiplication
({\tt mul}) is actually easier.

After making sure that the recursion is on its smaller argument,
{\tt mul} delegates its work to {\tt mul1} which uses the decomposition of
the first argument based on equation ({\bf \ref{cnat}}).
When the first term represents an even number, {\tt mul1} applies the 
{\tt leftshiftBy} operation corresponding to $2^{x+1}zy$, 
otherwise it derives a similar operation
from the representation $(2^{x+1}(z+1)-1)y$.
\begin{code}
mul :: Cat a => a -> a -> a
mul x y = f (cmp x y) where
  f GT = mul1 y x
  f _ = mul1 x y

mul1 :: Cat a => a -> a -> a  
mul1 x _ | e_ x = e
mul1 x y | u_ x = y
mul1 a y | even_ a =  
  leftshiftBy (s x) (mul1 z y) where (x,z) = c' a
mul1 a y | odd_ a = 
  sub (leftshiftBy (s x) (mul1 (s z) y)) y where (x,z) = c' a
\end{code}
Note that when the operands are composed of large blocks of alternating
{\tt 0} and {\tt 1} digits, the algorithm is quite efficient as it
works (roughly) in time depending on the the number of blocks in its
first argument rather
than its  number of digits.  
The following example illustrates a blend of 
arithmetic operations benefiting from complexity reductions
on giant tree-represented numbers:
\begin{codex}
*GCat> let term1 = 
    sub (exp2 (exp2 (t 12345))) (exp2 (t 6789))
*GCat> let term2 = 
    add (exp2 (exp2 (t 123))) (exp2 (t 456789))
*GCat> bitsize (bitsize (mul term1 term2))
C E (C E (C E (C (C E (C E E)) 
   (C (C E (C E (C E E))) (C (C E E) E)))))
*GCat> n it
12346
\end{codex}

\subsection{Power}

After specializing our multiplication for a  squaring operation,
\begin{code}
square x = mul x x
\end{code}
we can implement a simple but efficient
``power by squaring'' operation for $x^y$, as follows:
\begin{code}
pow _ x | e_ x = c (x,x)
pow a b | even_ a = c (s' (mul (s x) b),ys) where
  (x,xs) = c' a
  ys = pow xs b
pow a b | even_ b = pow (superSquare y a) ys where
  (y,ys) = c' b
  superSquare a x | e_ a = square x
  superSquare k x = square (superSquare (s' k) x)
pow x y = mul x (pow x (s' y))
\end{code}

It works well with fairly large numbers, by also
benefiting from efficiency of multiplication on terms
with large blocks of  {\tt 0} and {\tt 1} digits:
\begin{codex}
*GCat> n (bitsize (pow (t 10) (t 100)))
333
*GCat> pow (t 32) (t 10000000)
C (C (C E (C (C E E) E)) (C (C (C E E) (C E E)) (C (C (C E E) E) 
  (C E (C E (C E (C (C (C E E) (C E E)) (C E (C E E))))))))) (C E E)
\end{codex}


\subsection{General division}

An algorithm derived from the well known long division 
is  given here. While it
matches its bitstring equivalent asymptotically,
it does not provide the same complexity gains as,
for instance, multiplication, addition or subtraction. 
\begin{code}
div_and_rem :: Cat a => a -> a -> (a, a)
div_and_rem x y | LT == cmp x y = (e,x)
div_and_rem x y | c_ y  = (q,r) where 
  (qt,rm) = divstep x y
  (z,r) = div_and_rem rm y
  q = add (exp2 qt) z
\end{code}
The function {\tt divstep} implements a step of
the division operation.
\begin{code}    
  divstep n m = (q, sub n p) where
    q = try_to_double n m e
    p = leftshiftBy q m    
\end{code}
The function {\tt try\_to\_double} doubles its
second argument  while smaller than its first
argument and returns the number of steps it took.
This value will be used by {\tt divstep}
when applying the {\tt leftshiftBy} operation.
\begin{code}    
  try_to_double x y k = 
    if (LT==cmp x y) then s' k
    else try_to_double x (db y) (s k)   
\end{code}
By specializing {\tt div\_and\_rem} we
derive division and remainder.
\begin{code}
divide :: Cat b => b -> b -> b          
divide n m = fst (div_and_rem n m)

remainder :: Cat b => b -> b -> b
remainder n m = snd (div_and_rem n m)
\end{code}

The following examples illustrate the agreement with their usual counterparts:
\begin{codex}
*GCat> divide 26 3
8
*GCat> remainder 26 3
2
\end{codex}

\section{Specialized Arithmetic Operations and Primality Tests} \label{spec}

We describe in this section a number of special purpose arithmetic
operations showing the practical usefulness of our number representation.

\subsection{Integer Square Root}
A fairly efficient integer square root, using Newton's method,
is implemented as follows:
\begin{code}
isqrt x | e_ x = x
isqrt n = if cmp (square k) n==GT then s' k else k where
  two = c (e,(c (e,e)))
  k=iter n
  iter x = if cmp (absdif r x)  two == LT
    then r
    else iter r where r = step x
  step x = divide (add x (divide n x)) two
absdif x y = if LT == cmp x y then sub y x else sub x y
\end{code}

The following examples illustrate the agreement with their usual counterparts:
\begin{codex}
*GCat> isqrt (t 101)
C E (C E (C E (C E E)))
*GCat> n it
10
\end{codex}

\subsection{Modular Power}
The modular power operation $x^y (mod~m)$ can be optimized to avoid the creation
of large intermediate results, by combining ``power by squaring''
and pushing the modulo operation inside the inner function {\tt modStep}.
\begin{code}
modPow m base expo = modStep expo one base where
  one = c (e,e)
  modStep x r b | x == one = (mul r b) `remainder` m
  modStep x r b | odd_ x =
    modStep (hf (s' x)) (remainder (mul r b) m)
                        (remainder (square b)  m)
  modStep x r b = modStep (hf x) r (remainder (square b) m)
\end{code}

The following examples illustrate the correctness of these operations:
\begin{codex}
*GCat> modPow (t 10) (t 3) (t 3)
C (C E (C E E)) E
*GCat> n it
7
\end{codex}

\subsection{Lucas-Lehmer Primality Test for Mersenne Numbers}
The Lucas-Lehmer primality test
has been used for the discovery of all the
record holder largest known prime numbers of the form $2^p-1$ with $p$ prime
in the last few years. 
It is based on iterating $p-2$ times the function $f(x)=x^2-2$, starting from $x=4$.
Then $2^p-1$ is prime if and only if the result modulo $2^p-1$ is $0$, as proven
in \cite{bruce93}.
The function {\tt ll\_iter} implements this iteration.
\begin{code}
ll_iter k n m |e_ k = n
ll_iter k n m = fastmod y m where
   x = ll_iter (s' k) n m
   y = s' (s' (square x))
\end{code}
It relies on the function {\tt fastmod} which provides
a specialized fast computation of $k~modulo~(2^p-1)$.
\begin{code}
fastmod k m | k == s' m = e
fastmod k m | LT == cmp k m = k
fastmod k m = fastmod (add q r) m where
  (q,r) = div_and_rem k m
\end{code}
Finally the Lucas-Lehmer primality test is implemented as follows:
\begin{code}
lucas_lehmer p | p == s (s e) = True
lucas_lehmer p = e ==  (ll_iter p_2 four m) where
  p_2 = s' (s' p)
  four = c (c (e,e),c (e,e))
  m  = exp2 p
\end{code}
We illustrate its use for detecting a few Mersenne primes:
\begin{codex}
*GCat> map n (filter lucas_lehmer (map t [3,5..31]))
[3,5,7,13,17,19,31]
*GCat> map (\p->2^p-1) it
[7,31,127,8191,131071,524287,2147483647]
\end{codex}
Note that the last line contains the Mersenne primes corresponding to  $2p+1$.

 \subsection{Miller-Rabin Probabilistic Primality Test}
Let $\nu_2(x)$ denote the {\em dyadic valuation of x}, i.e., the largest
exponent of 2 that divides x. The function {\tt dyadicSplit}
defined by equation (\ref{dia})
\begin{equation} \label{dia}
dyadicSplit(k) = {(k,{k \over 2^{\nu_2(k)}})}
\end{equation}
can be implemented as an average constant time operation as:
\begin{code}
dyadicSplit z | odd_ z = (e,z)
dyadicSplit z | even_ z = (s x, s xs) where
  (x,xs) = c' (s' z)
\end{code}

After defining a sequence of {\tt k} random natural numbers in an interval
\begin{code}
randomNats :: (Cat a) => Int -> Int -> a -> a -> [a]
randomNats seed k smallest largest =  map view ns  where
  ns = take k (randomRs
    (n smallest,n largest) (mkStdGen seed))
\end{code}

we are ready to implement the function {\tt oddPrime} that
runs {\tt k} tests and concludes primality
with probability $1-{1 \over 4^k}$ if all {\tt k}
calls to function {\tt strongLiar} succeed.
\begin{code}
oddPrime :: Cat a => Int -> a -> Bool
oddPrime k m = all strongLiar as where
  one = s e
  two = s one
  m' = s' m
  as = randomNats k k two m'
  (l,d) = dyadicSplit m'

  strongLiar a = (x == one || (any (== m') (squaringSteps l x))) where
    x = modPow m a d

    squaringSteps x _ | e_ x = []
    squaringSteps l x = x:squaringSteps (s' l)
      (remainder (square x) m)
\end{code}
Note that we use {\tt dyadicSplit} to find a pair {\tt (l,d)} such that {\tt l}
is the largest power of {\tt t 2} dividing
the predecessor {\tt m'} of the suspected prime {\tt m}.
The function {\tt strongLiar} checks, for a random base {\tt a}, a
condition that primes (but possibly also a few composite numbers)
verify.

Finally {\tt isProbablyPrime} handles the case of even numbers and
calls {\tt oddPrime} with the parameter specifying the number of tests,
{\tt k=42}.
\begin{code}
isProbablyPrime x | x==two  = True where two = s (s e)
isProbablyPrime x | even_ x = False
isProbablyPrime p = oddPrime 42 p
\end{code}
The following example illustrates the correct behavior of the
algorithm on a the interval {\tt [2..100]}.
\begin{codex}
*GCat> map n (filter isProbablyPrime (map t [2..100]))
[2,3,5,7,11,13,17,19,23,29,31,37,41,43,47,53,59,61,67,71,73,79,83,89,97]
\end{codex}

\section{Boolean Operations on Tree-represented Bitvectors} \label{bool}
We  implement bitvector operations (also seen
as efficient bitset operations) to work ``one block of
binary digits at a time'', to facilitate
their use on large but sparse boolean formulas
involving a large numbers of variables. One will be able
to evaluate such formulas ``all value-combinations at a time''
when represented as bitvectors of size $2^{2^n}$.
Note that such operations will be tractable with
our trees, provided that they have a
relatively small
representation complexity,
despite their large bitsize.

\subsection{Bitwise Operations One Block of Digits at a time}
We implement a generic {\tt bitwise} operation that takes
a boolean function {\tt bf} as its first parameter.

First, when an argument is {\tt F []}, corresponding to {\tt 0}
the behavior is derived from that of the boolean function {\tt bf}.
\begin{code}
bitwise :: (Cat a) => (Bool -> Bool -> Bool) -> a -> a -> a
bitwise bf x y | e_ x && e_ y = x
bitwise bf x y | e_ x = if bf False True then y else x
bitwise bf x y | e_ y = if bf True False then x else y
\end{code}

Next, the parities of the arguments {\tt px} and {\tt py}
are used to derive the parity of the result {\tt pz}, by applying
the boolean function {\tt bz}.
\begin{code}
bitwise bf x y  = f (cmp a b) where
  (a,as) = c' x
  (b,bs) = c' y
  px = odd_ x
  py = odd_ y
  pz = bf px py
\end{code}
Based on the parity {\tt pz} the local function {\tt f}
(also parameterized by the result of the comparison between
arguments {\tt x} and {\tt y}) is called.
\begin{code}
  f EQ = fApply bf pz (s a) as bs
  f GT = fApply bf pz (s b) (fromB px (sub a b) as) bs
  f LT = fApply bf pz (s a) as (fromB py (sub b a) bs)
\end{code}
The function {\tt f} calls {\tt fromB} to derive from the parities
{\tt px} and {\tt py} the appropriate left-shifting operation.
\begin{code}
  fromB False = leftshiftBy
  fromB True  = leftshiftBy'
\end{code}

Finally, the function {\tt f} calls the helper function {\tt fApply},
which, depending on the expected parity of the result {\tt pz}, applies the appropriate
left-shift operation to the result of the recursive application of {\tt bitwise} to
the remaining blocks of digits {\tt u} and {tt v}.
\begin{code}
  fApply bf False k u v =  leftshiftBy k (bitwise bf u v)
  fApply bf True k u v =  leftshiftBy' k (bitwise bf u v)
\end{code}

The actual bitwise operations are obtained by parameterizing the generic {\tt bitwise}
function with the appropriate Haskell boolean functions:
\begin{code}
bitwiseOr x y= bitwise (||) x y

bitwiseXor  x y = bitwise (/=)  x y

bitwiseAnd  x y = bitwise (&&)  x y
\end{code}

Bitwise negation (requiring the additional parameter {\tt k}
to specify the intended bitlength of the operand)
corresponds to the complement w.r.t. the
``universal set'' of all
natural numbers up to $2^k-1$.
It is defined as usual, by subtracting from
the ``bitmask'' corresponding to $2^k-1$:
\begin{code}
bitwiseNot k x = sub y x where y = s' (exp2 k)
\end{code}
The function {\tt bitwiseAndNot}
combines {\tt bitwiseOr}, {\tt bitwiseNot}
the usual way, except that it uses the helper
function {\tt bitsOf} to ensure enough mask bits
are made available when negation is applied.
\begin{code}
bitwiseAndNot x y = bitwiseNot l d  where
  l = max2 (bitsOf x) (bitsOf y)
  d = bitwiseOr (bitwiseNot l x) y
\end{code}
The function {\tt bitsOf} adapts our definition
for {\tt bitsize} to compute the number of bits of a bitvector (considering
0 to be 1 bit).
\begin{code}
bitsOf x | e_ x = s x
bitsOf x = bitsize x
\end{code}
The following example illustrates that
our bitwise operations can be efficiently applied
to giant numbers:
\begin{codex}
*GCat> bitwiseXor (s (exp2 (exp2 (m 12345))))  (s' (exp2 (exp2 (m 6789))))
F [F [],F [F [],F [F [F []],F [F [F []],F []],F [],F [],F [],F [],
  F [F []]]],F [F [F [F []],F [],F [F [F []]],F [],F [],F [],F [],
  F [F []]],F [],F [F [],F [],F [F []],F [F []],F [],F [F []],F [],
  F [],F [],F []]],F []]
\end{codex}
Note that while the size of the term representing
this result is {\tt 46} {\tt F} nodes the bitsize of
the result is $2^{12346}$ showing clearly that such an operation is intractable
with a bitstring representation.

\subsection{Boolean Formula Evaluation}
Besides definitions for the bitwise boolean functions, we also need
definitions of the projection variables
$var(n,k)$ corresponding to column $k$ of a truth table, for a function with $n$
variables.
A compact formula for them, as given in \cite{knuth_boolean} or \cite{synasc12b}, is
\begin{equation}
var(n,k)=({{2^{2^n}-1}) ~/~ ({2^{2^{n-k-1}}+1}})
\end{equation}
However, instead of doing the division, one can compute them
as a concatenation of alternating blocks of $1$ and $0$ bits
to take advantage of our efficient block operations.
\begin{code}
var n k = repeatBlocks nbBlocks blockSize mask where
  k' = s k
  nbBlocks = exp2 k'
  blockSize = exp2 (sub n k')
  mask = s' (exp2 blockSize)
\end{code}
The alternating blocks are put together by the function
{\tt repeatBlocks} that shifts to the left
by the size of a block, at each step, and
adds the {\tt mask}
made of $2^{n-k}$ ones, at each even step.
\begin{code}
  repeatBlocks x _ _ | e_ x = x
  repeatBlocks k l mask =
   if odd_ k then r else add mask r where
    r = leftshiftBy l (repeatBlocks (s' k) l mask)
\end{code}

The following examples illustrate these operations:
\begin{codex}
*GCat> map n (map (var (t 3)) (map t [0..2]))
[15,51,85]
*GCat> map n (map (var (t 4)) (map t [0..3]))
[255,3855,13107,21845]
*GCat> map n (map (var (t 5)) (map t [0..4]))
[65535,16711935,252645135,858993459,1431655765]
\end{codex}

The following  example illustrates the evaluation of
a boolean formula in conjunctive normal form (CNF).
The mechanism is usable as a simple satisfiability
or tautology tester, for formulas resulting in
possibly large but sparse or
dense, low structural complexity
bitvectors.
\begin{code}
cnf tf = andN (map orN cls) where
  cls = [[v0',v1',v2],[v0,v1',v2],
         [v0',v1,v2'],[v0',v1',v2'],[v0,v1,v2]]

  v0 = var (tf 3) (tf 0)
  v1 = var (tf 3) (tf 1)
  v2 = var (tf 3) (tf 2)

  v0' = bitwiseNot (exp2 (tf 3)) v0
  v1' = bitwiseNot (exp2 (tf 3)) v1
  v2' = bitwiseNot (exp2 (tf 3)) v2

  orN (x:xs) = foldr bitwiseOr x xs
  andN (x:xs) = foldr bitwiseAnd x xs
\end{code}
The execution of the function {\tt cnf} evaluates the formula,
the result corresponding to bitvector {\tt 88 = [0,0,0,1,1,0,1,0]}.
\begin{codex}
*GCat> cnf t
C (C E (C E E)) (C (C E E) (C E (C E E)))
*GCat> cnf m
F [F [F [],F []],F [F []],F [],F []]
*GCat> cnf n
88
\end{codex}

\section{Representing Sets}

\subsection{The Bijection between sequences and sets}

Increasing sequences provide a canonical representation for
sets of natural numbers.
While finite sets and finite lists of elements of $\N$
share a common representation {\tt [N]}, sets 
are subject to the implicit constraint that their
ordering is immaterial.
This suggest that a set like $[4,1,9,3]$ could be
represented canonically as a sequence by first ordering it as $[1,3,4,9]$ and
then computing the differences between consecutive elements i.e.
$[x_0,x_1 \ldots x_i, x_{i+1} \ldots] \rightarrow [x_0,x_1-x_0, \ldots
x_{i+1}-x_i \ldots]$.
This gives $[1,2,1,5]$, with
the first element $1$ followed by the increments $[2,1,5]$

Therefore, {\em incremental sums}, computed with Haskell's {\tt scanl}, are used to transform arbitrary
lists to canonically represented
sets of natural numbers, inverted by {\em pairwise differences} computed
using {\tt zipWith}.
\begin{code}   
list2set [] = []
list2set (n:ns) = scanl add n ns

set2list [] = []
set2list (m:ms) = m : zipWith sub ms (m:ms)
\end{code}

\subsection{The bijections between natural numbers and sets of natural numbers}
By composing with natural number-to-list bijections, we obtain bijections
to sets of natural numbers.
\begin{code}
to_set x = list2set (to_list x)

from_set x = from_list (set2list x)
\end{code}
As the following example shows, trees 
offer a significantly more compact
representation of sparse sets than conventional
binary numbers.
\begin{codex}
*GCat> n (bitsize (from_set (map t [42,1234,6789])))
6792
*GCat> n (catsize (from_set (map t [42,1234,6789])))
33
\end{codex}
Note that a similar compression occurs for sets of natural numbers 
 with only a few elements missing (that we call {\em dense sets}), as
they have the same representation 
size as their sparse complement.
\begin{codex}
*GCat> n (catsize (from_set (map t ([1,3,5]++[6..220]))))
438
*GCat> n (bitsize (from_set (map t ([1,3,5]++[6..220]))))
438
\end{codex}
The following holds:
\begin{prop}
These encodings/decodings of lists and sets  as trees
are size-proportionate i.e., their representation sizes are within constant factors.
\end{prop}

\section{Representation complexity} \label{stru}

While a precise average complexity analysis of our algorithms is beyond the scope of this paper, arguments similar to those about the average behavior of {\tt s} and {\tt s'} can be carried out to prove that for all our operations, {\em their average complexity matches their traditional counterparts}, using the fact, shown in the proof of Prop. \ref{avg}, 
that the average size of a  block of contiguous {\tt 0} or {\tt 1} 
bits is at most {\tt 2}.

\subsection{Complexity as representation size}
To evaluate the best and worst case space requirements
of our number representation, 
we introduce here a measure of {\em representation complexity}, defined
by the function {\tt catsize} that counts the non-empty nodes of an object of
type $\C$.
\begin{code}
catsize :: Cat a => a -> a
catsize z | e_ z = z
catsize  z = s (add (catsize x) (catsize y)) where (x,y) = c' z
\end{code}
The following holds:
\begin{prop} \label{bitcmp}
For all terms $t \in \C$, {\tt catsize t} $\leq$ {\tt bitsize t}.
\end{prop}
\begin{proof}
By induction on the structure of $t$, observing that the two
functions have similar definitions and corresponding calls to
{\tt catsize} return terms inductively assumed smaller than those of
{\tt bitsize}.
\end{proof}

The following example illustrates their use:
\begin{codex}
*GCat> map catsize [0,100,1000,10000]
[0,7,9,13]
*GCat> map catsize [2^16,2^32,2^64,2^256]
[5,6,6,6]
*GCat> map bitsize [2^16,2^32,2^64,2^256]
[17,33,65,257]
\end{codex}
Figure \ref{tsizes} shows the reductions in representation
complexity compared with bitsize for an initial interval of $\N$,
from $0$ to $2^{10}-1$.
\FIG{tsizes}{Representation complexity (lower line) bounded by bitsize (upper line)} {0.25}{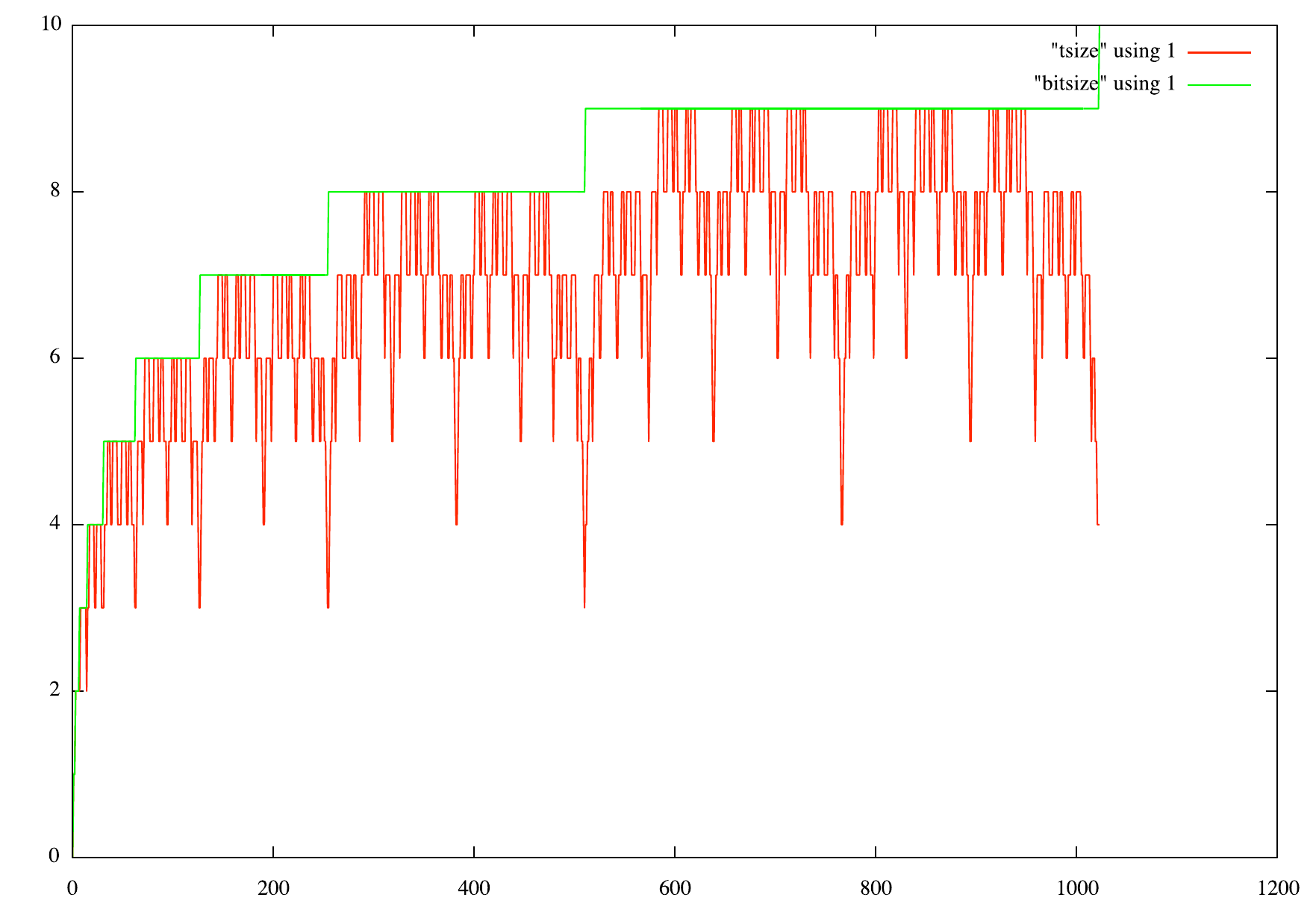}

\subsection{Enumerating objects of given representation size}

The total number of Catalan objects corresponding to $n$ is given by:
\begin{equation} \label {cn}
C_n = {1 \over {n+1}}{{2n} \choose n} 
\end{equation}
It is shown in \cite{stanleyEC} that if
\begin{math}
L_n = {{2^{2n}} \over {n ^{3 \over 2}}\sqrt{\pi}}
\end{math}
then
\begin{math}
{\lim_{n \to \infty}{C_n \over L_n}}=1
\end{math} , providing an asymptotic bound for $C_n$.

The function {\tt cat} describes an
efficient computation for of the Catalan number $C_n$
using a direct recursion formula derived from equation (\ref{cn}).
\begin{code}
cat :: N->N
cat 0 = 1
cat n | n>0 = (2*(2*n-1)*(cat (n-1))) `div` (n+1)
\end{code}
The first few members of the sequence are computed below:
\begin{codex}
*GCat> map cat [0..14]
[1,1,2,5,14,42,132,429,1430,4862,16796,58786,208012,742900,2674440]
\end{codex}
The following holds.
\begin{prop}
Let {\tt k = catsize x} where {\tt x} is an object of type $\C$. Then {\tt x}
relates to {\tt k} as follows, for instances of $\C$:
\begin{enumerate}
\item {\tt x} is a binary tree of type $\T$ with $k$ internal nodes and $k+1$ leaves
\item {\tt x} is a multiway tree of type $\M$ with $k+1$ nodes
\item {\tt x} is a term  of the type $\P$ with $2k+2$ matching parentheses.
\end{enumerate}
\end{prop}
\begin{proof}
Observe that {\tt catsize k} counts the $C_{\tt k}-1$ 
number of {\tt C} constructors in objects
of size {\tt k} of type {\tt T}.
The rest is a consequence of well-known relations between Catalan numbers
and  nodes of binary trees, nodes of multiway trees and parentheses
in Dyck words as given in \cite{stanleyEC}.
\end{proof}

The function {\tt catsized}
enumerates for each of our instances, the objects of
size {\tt k} corresponding to a given Catalan number.
\begin{code}
catsized :: Cat a => a -> [a]
catsized a = take k [x | x<-iterate s e,catsize x == a] where
  k = fromIntegral (cat (n a))
\end{code}
The function extracts exactly {\tt k} elements (with {\tt k} 
the Catalan number corresponding to size {\tt a})
from the infinite
enumeration of Catalan objects of type {\tt Cat} 
provided by {\tt iterate s e},
as illustrated below:
\begin{codex}
*GCat> catsized (t 2)
[C E (C E E),C (C E E) E]
*GCat> catsized 4
[8,9,10,11,12,13,14,16,30,31,63,127,255,65535]
\end{codex}

\subsection{Best and worst cases}
Next we define the higher order function {\tt iterated} 
that applies {\tt k} times the function {\tt f}, which, contrary
to Haskell's {\tt iterate}, returns only the final element rather than
building the infinite list of all iterates.
\begin{code}
iterated :: Cat a => (t -> t) -> a -> t -> t 
iterated f k x |e_ k = x
iterated f k x = f (iterated f (s' k) x) 
\end{code}
We can exhibit, for a given bitsize, a best case
\begin{code}
bestCase :: Cat a => a -> a
bestCase k = iterated f k e where f x = c (x,e)
\end{code}
and a worst case
\begin{code}
worstCase :: Cat a => a -> a
worstCase k = iterated f k e where f x = c (e,x)
\end{code}
The following examples illustrate these functions:
\begin{codex}
*GCat> bestCase (t 5)
C (C (C (C (C E E) E) E) E) E
*GCat> n (bitsize (bestCase (t 5)))
65536
*GCat> n (catsize (bestCase (t 5)))
5
*GCat> worstCase (t 5)
C E (C E (C E (C E (C E E))))
*GCat> n (bitsize (worstCase (t 5)))
5
*GCat> n (catsize (worstCase (t 5)))
5
\end{codex}
The function {\tt bestCase } computes the iterated
exponent of 2 (tetration) and then applies the predecessor
to it. For $k=4$ it corresponds to\\\\
 ${(2^{{{(2^{{{(2^{{{(2^{{0+1}}-1)}+1}}-1)}+1}}-1)}+1}}-1)}=2^{2^{2^2}}-1=65535$.\\\\
For $k=5$ it corresponds to $2^{65536}-1$.
 
Note that our concept of representation complexity is only
a weak approximation
of Kolmogorov complexity \cite{vitanyi}.
For instance, the reader might
notice that  our worst case example
is computable by a program of relatively
small size. However, as {\tt bitsize}
is an upper limit to {\tt catsize}, we can
be sure that we are within constant
factors from the corresponding bitstring
computations, even on random data
of high Kolmogorov complexity.

Note also that an alternative concept of representation
complexity can be defined by considering the (vertices+edges) size of the DAG obtained
by folding together identical subtrees.

\subsection{A Concept of duality}
As our instances of {\tt Cat} are
members of the Catalan family of combinatorial objects,
they can be seen as binary trees with empty leaves.
Therefore, we can transform the tree representation of our
objects by swapping left and right branches under a binary tree view,
recursively.
The corresponding Haskell code is:
\begin{code}
dual :: Cat a => a -> a
dual x | e_ x = e
dual z = c (dual y,dual x) where (x,y) = c' z
\end{code}
As clearly {\tt dual} is an {\em involution} (i.e., {\tt dual}~$\circ$~{\tt dual} is the identity of $\C$), the corresponding permutation of $\N$
will bring together huge and small natural numbers sharing 
representations of the same size, as illustrated by the following example.
\begin{codex}
*GCat> map dual [0..20]
[0,1,3,2,4,15,7,6,12,31,65535,16,8,255,127,5,11,8191,4294967295,32,65536]
*CatsBM> catShow 10
"(()()()())"
*CatsBM> catShow (dual 10)
"((((()))))"
\end{codex} 
For instance, {\tt 18} and its dual {\tt 4294967295} have
representations of the same size,
except that the wide tree associated to {\tt 18} maps to the tall tree associated to
{\tt 4294967295}, as illustrated by Fig. \ref{tall}, with trees
folded to DAGs by merging together shared subtrees. Note the significantly different
bitsizes that can result for a term and its dual.
\HFIGS{tall}
{18 and its dual, with multiway trees folded to DAGs} {18}{dual of 18}{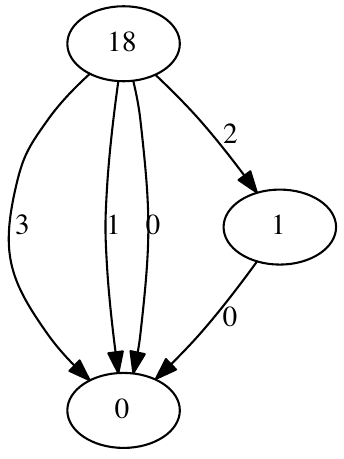}{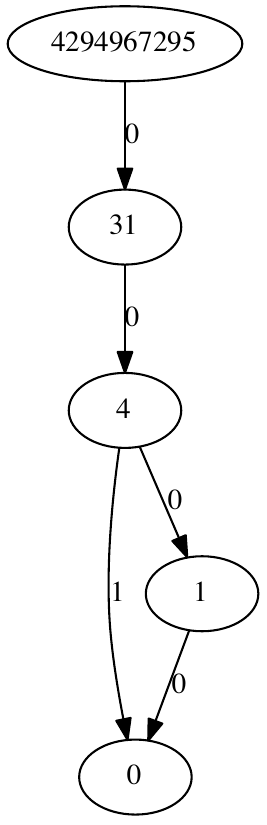}{0.50}
\begin{codex}
*GCat> m 18
F [F [],F [],F [F []],F []]
*GCat> dual (m 18)
F [F [F [F [F []],F []]]]
*GCat> n (bitsize (m 18))
5
*GCat> n (bitsize (dual (m 18)))
32
\end{codex}
It follows immediately from the definitions of the respective functions, 
that, as an extreme case, the following holds:
\begin{prop}
{\tt $\forall${\tt x}. dual (bestCase x)} = {\tt worstCase x}.
\end{prop}

The following example illustrates it, with 
a tower of exponents {\tt 10000} tall, reached
by {\tt bestCase}. Note that we run it on objects of type {\tt T}, as it would
dramatically overflow memory on bitstring-represented numbers of type $\N$.
\begin{codex}
*GCat> dual (bestCase (t 10000)) == worstCase (t 10000)
True
\end{codex}
Another interesting property of {\tt dual} is illustrated by the following examples:
\begin{codex}
*GCat> [x|x<-[0..2^5-1],cmp (t x) (dual (t x)) ==  LT]
[2,5,6,8,9,10,11,13,14,17,18,19,20,21,22,23,25,26,27,28,29,30]
*GCat> [x|x<-[0..2^5-1],cmp (t x) (dual (t x)) ==  EQ]
[0,1,4,24]
*GCat> [x|x<-[0..2^5-1],cmp (t x) (dual (t x)) ==  GT]
[3,7,12,15,16,31]
\end{codex}
The discrepancy between the number of elements for which {\tt x} is smaller than
{\tt (dual x)} and those for which it is greater or equal, is growing
as numbers get larger,
contrary to the intuition that, as {\tt dual} is an {\em involution}, the grater and smaller
sets would have similar sizes for an initial interval of $\N$. 
For instance, {\em between $0$ and $2^{16}-1$ one will find only $68$ numbers for which the dual is smaller and $11$ for which it is equal}.

Note that random elements of $\N$ tend to have relatively {\em shallow} (and wide)
multiway tree representations, given
that the average size of a contiguous block of {\tt 0}s or {\tt 1}s is {\tt 2}.
Consequently, 
{\tt dual} provides an interesting bijection between ``incompressible'' natural numbers (of high Kolmogorov complexity) and their {\em deep}, highly compressible, duals.

The existence of such a bijection (computed by a program of constant size)
between natural numbers of high and low Kolmogorov complexity reveals a
somewhat non-intuitive aspect of this concept and its use for the
definition of randomness \cite{vitanyi}.

We will explore next definitions for concepts of depth for our
number representation.

\subsection{Representation Depth} \label{repdep}
As we can switch between the binary and multiway view of our
Catalan objects, we will define two sets of representation-depth 
functions. They use the the helper 
functions  minimum {\tt min2} and maximum {\tt max2}.
\begin{code}
min2, max2 :: Cat a => a -> a -> a
min2 x y = if LT==cmp x y then x else y
max2 x y = if LT==cmp x y then y else x
\end{code}
Corresponding to the {\em binary tree view} exemplified by instance {\tt T},
we define {\tt maxTdepth} 
returning the length of
the longest
path from the root to a leaf.
\begin{code}
maxTdepth :: Cat a => a -> a
maxTdepth z | e_ z = z
maxTdepth z = s (max2  (maxTdepth x) (maxTdepth y)) where (x,y) = c' z
\end{code}
\begin{codeh}
minTdepth :: Cat a => a -> a
minTdepth z | e_ z = z
minTdepth z = s (min2  (minTdepth x) (minTdepth y)) where (x,y) = c' z  
\end{codeh}

Corresponding to the {\em multiway tree view} exemplified by instance {\tt M}
we define {\tt maxTdepth} 
returning the length of
the longest
path from the root to a leaf.
\begin{code}
maxMdepth :: Cat a => a -> a
maxMdepth z | e_ z = z
maxMdepth z = s (foldr max2 m ms) where
  (m:ms) = map maxMdepth (to_list z)
\end{code}  
\begin{codeh}
minMdepth :: Cat a => a -> a
minMdepth z | e_ z = z
minMdepth z = s (foldr min2 m ms) where
  (m:ms) = map minMdepth (to_list z)
\end{codeh}

The following simple facts hold, derived from properties of binary and multiway rooted ordered trees.
\begin{prop}
Let ${\tt x} \geq {\tt y}$ stand for {\tt cmp x y == GT} and $=$ stand for {\tt cmp x y == EQ}.
\begin{enumerate}
\item For all objects {\tt x} of type $\C$, {\tt catsize x} $\geq$ {\tt maxTdepth x} $\geq$ {\tt maxMdepth x}.   

\item For all objects {\tt x} of type $\C$, {\tt catsize x} = {\tt catsize (dual x)} 
\item
For all objects {\tt x} of type $\C$ 
{\tt maxTdepth x} = {\tt maxTdepth (dual x)}. 

\item For all objects {\tt x} of type $\C$, {\tt maxMdepth (bestCase x)} $=$ {\tt x}.
\end{enumerate}
\end{prop}

\begin{codeh}
-- equivalent to maxMdepth

maxMdepth' :: Cat a => a -> a
maxMdepth' z | e_ z = z
maxMdepth' z = s (max2 (maxMdepth' x) y') where 
  (x,y) = c' z
  y' = if c_ y then maxMdepth' (snd (c' y)) else e

rratio x = fromIntegral md / fromIntegral td where
  md = maxMdepth x
  td = maxTdepth x 
  
rdif x = fromIntegral td - fromIntegral md where
  md = maxMdepth x
  td = maxTdepth x   
\end{codeh}

\section{Compact representation and tractable computations with some giant numbers} \label{giant}

We will illustrate the representation and computation power of our new numbering system
with two case studies. The first shows that several {\em record holder primes} have
compact Catalan representations.

The second shows computation of the hailstone sequence for an equivalent of the {\em Collatz conjecture} on giant  numbers.

\subsection{Record holder primes}

Interestingly, most record holder giant primes have a fairly simple structure:
they are of the form $p=i2^k \pm j$ with $i \in \N$ and $j \in \N $ comparatively small.
This is a perfect fit for our representation which favors numbers ``in the
neighborhood'' of linear combinations of (towers of) exponents of two with 
comparatively small coefficients, resulting in large contiguous blocks
of $0$s and $1$s when represented as bitstrings.  

The largest known primes (as of early 2019) of several types are given by the following Haskell code.
\begin{code}
mersennePrime f = s' (exp2 (f 82589933))
generizedFermatPrime f = s (leftshiftBy (f 9167433) (f 27653))
cullenPrime f = s (leftshiftBy x x) where x = f 6679881
woodallPrime f = s' (leftshiftBy x x) where x = f 3752948
prothPrime f = s (leftshiftBy (f 13018586) (f 19249))
sophieGermainPrime f = s' (leftshiftBy (f 666667) (f 18543637900515))
twinPrimes f = (s' y,s y) where 
  y = leftshiftBy (f 666669) (f 3756801695685)
\end{code}

\begin{codeh}  

giants :: Cat a => (N -> a) -> ([String], [a])
giants f = (ns,ps) where
  ps = [mersennePrime f, generizedFermatPrime f,  cullenPrime f, woodallPrime f,
        sophieGermainPrime f, fst (twinPrimes f), snd (twinPrimes f)]
  ns = ["mersenne48", "generizedFermatPrime",  "cullenPrime", "woodallPrime",
        "sophieGermainPrime", "twinPrimes1", "twinPrimes2"]
        
-- sizes for primes and their  duals      
sizes d f =  zip (zip3 ns bs cs) (zip3 maxs tmaxs dmaxs) where 
  ps = map d (snd (giants f))
  ns = fst (giants f)
  
  bs = map (n.catsize) ps
  cs = map (n.catsize) ps
  maxs = map (n.maxMdepth) ps
  tmaxs = map (n.maxTdepth) ps
  dmaxs = map (n.maxMdepth.dual) ps

-- usage: showSizes t
showSizes f = mapM_ print (sizes id f)
showDSizes f = mapM_ print (sizes dual f)

-- usage: compDuals t
compDuals f = zip ns (zipWith cmp ps ds) where
  ps = snd (giants f)
  ds = map dual ps
  ns = fst (giants f)

ilog2star :: Cat a => a->a
ilog2star x | e_ x = x
ilog2star x = s (ilog2star (ilog2 x))

logsizes =  map (n.ilog2star) (snd (giants t))

dlogsizes = map (n.ilog2star) (map dual (snd (giants t)))

maxmdepths = map n (map maxMdepth (snd (giants t)))
maxmDdepths = map n (map maxMdepth (map dual (snd (giants t))))

maxTdepths = map n (map maxTdepth (snd (giants t)))
maxTDdepths = map n (map maxTdepth (map dual (snd (giants t))))
\end{codeh}

For instance, the largest known prime, having about 17 million decimal digits,
(a Mersenne number) has an unusually small Catalan representation
as illustrated below:
\begin{codex}
*GCat> catShow (mersennePrime t)
"(((())(())()(()())(()())(()())((()))(())()(()())(())()))"
*GCat> n (catsize (mersennePrime t))
27
*GCat> n (bitsize (mersennePrime t))
82589933
\end{codex}
Note the use of parameter {\tt t} 
indicating that computation proceeds with type {\tt T},
as it would overflow memory with with bitstring-represented natural numbers.

Figure \ref{giantprimes} summarizes comparative bitstring and Catalan representation
sizes {\tt bitsize} and {\tt catsize} for record holder primes.
 
\begin{figure}
\begin{center}
\begin{tabular}{||l|r|r||}
\hline
\multicolumn{1}{||c|}{{Record holder prime}} & 
  \multicolumn{1}{c|}{{\tt bitsize}} &
  \multicolumn{1}{c||}{{\tt catsize}}\\
\hline 
\hline
Mersenne prime & 82,589,933 & 27 \\ \hline
Generalized Fermat prime& 9,167,448 & 37 \\ \hline
Cullen prime & 6,679,904 & 46   \\ \hline
Woodall prime & 3,752,970 & 37 \\ \hline
Sophie Germain prime & 666,712 & 62 \\ \hline
Twin primes 1 & 666,711 & 59 \\ \hline
Twin primes 2 & 666,711& 60 \\ \hline
\end{tabular} \\
\medskip
\caption{Bitsizes vs. Catalan representation sizes of record holder primes
\label{giantprimes}}
\end{center}
\end{figure}

\begin{codex}
*GCat> catShow (genFermatPrime t)
"(()((()())(()())()(())()(()())((())())()()(()())())
  ()()()(()(()))(())()(()))"
*GCat> catShow (cullenPrime t)
"(()((()())(()())()()()()(())()((()))()()(())(()))()
 (())()(())()()()()(())()((()))()()(())(()))"
*GCat> catShow (woodallPrime t)
"((()()()()(()()())((()))()()()(())(()()))(())(()()())
 ((()))()()()(())(()()))"
*GCat> catShow (sophieGermainPrime t)
"((()()()()()()((()))(())()()(()())()()())()()(()())(()())(())
 (())()()(())(())(())((()))()(())((()))(())()(()())()(())((()))())"
*GCat> catShow (fst (twinPrimes t))
"(((())(())()()((()))(())()()(()())()()())(())()(()())(()(()))(())
 ()(())()()()()(())()()(())(())()()()()()()()(())()(()))"
*GCat> catShow (snd (twinPrimes t))
"(()((())()()()()((()))(())()()(()())()()())()()()(()())(()(()))(())
 ()(())()()()()(())()()(())(())()()()()()()()(())()(()))"
*GCat> catShow (snd (mersenne48 t))

*GCat> catsize (dual (sophieGermainPrime t))
C E (C (C (C E E) (C E E)) E)
*GCat> n it
62
*GCat> catsize (dual (prothPrime t))
C E (C (C E E) (C E (C E (C E E))))
*GCat> n it
41
*GCat> catsize (add (dual (prothPrime t)) (dual (sophieGermainPrime t)))
C (C E E) (C E (C E (C E (C (C E E) (C (C E E) E)))))
*GCat> n it
404
\end{codex}

\subsection{Computing the Collatz/Syracuse sequence for huge numbers} \label{collatz}

As an interesting application, that achieves something one cannot
do with traditional arbitrary bitsize integers is to explore the behavior of interesting
conjectures in the ``new world'' of numbers limited not by their
sizes but by their representation complexity.
The Collatz conjecture \cite{arxiv:lagarias} states that the function
\begin{equation}
collatz(x)=
\begin{cases}
x \over 2  &  \text{if $x$ is even},\\
3x+1 &  \text{if $x$ is odd}.\\
\end{cases}
\end{equation}
reaches $1$ after a finite number of iterations.
An equivalent formulation, by grouping together all
the division by 2 steps, is the function:
\begin{equation}
collatz'(x)=
\begin{cases}
x \over 2^{\nu_2(x)}  &  \text{if $x$ is even},\\
3x+1 &  \text{if $x$ is odd}.\\
\end{cases}
\end{equation}
where $\nu_2(x)$ denotes the {\em dyadic valuation of x}, i.e., the largest
exponent of 2 that divides x. One step further, the
{\em syracuse function} is defined as the odd integer $k'$ such
that $n=3k+1=2^{\nu_2(n)}k'$. One more step further, by writing $k'=2m+1$
we get a function that associates $k \in \N$ to $m \in \N$.

The function {\tt tl} computes efficiently the equivalent of 
\begin{equation}
tl(k) = {{{k \over 2^{\nu_2(k)}}-1} \over 2}
\end{equation}
Together with its {\tt hd} counterpart, it is defined as
\begin{code}
hd x = fst (decons x)

tl x = snd (decons x)

decons a | even_ a = (s x,hf (s' xs)) where (x,xs) = c' a
decons a = (e,hf (s' a))
\end{code}
where the function {\tt decons} is the inverse of 
\begin{code}
cons (x,y) = leftshiftBy x (s (db y))
\end{code}
corresponding to $2^x~(2y+1)$.
Then our variant of the {\em syracuse function} corresponds to
\begin{equation}
syracuse(n) = tl(3n+2)
\end{equation}
which is defined from $\N$ to $\N$ and can be implemented
as
\begin{code}
syracuse :: Cat b => b -> b
syracuse n = tl (add n (db (s n)))
\end{code}
Note that all operations except the addition {\tt add} are constant average time. 

The function {\tt nsyr} computes the iterates of this function, 
until (possibly) stopping:
\begin{code}
nsyr :: Cat t => t -> [t]
nsyr x | e_ x = [e]
nsyr x = x : nsyr (syracuse x)
\end{code}

It is easy to see that the Collatz conjecture is true\footnote{
As  a side note, it might be interesting to approach the Collatz conjecture
by trying to compare the growth in the {\em Catalan representation size}
induced by {\tt 3n+2} expressed as {\tt add n (db (s n))} vs. its
decrease induced by {\tt tl}.
}
if and only if {\tt nsyr}
terminates for all $n$, as illustrated by the following example:
\begin{codex}
*GCat> nsyr 2019
[2019,3029,4544,3408,2556,1917,2876,2157,3236,2427,3641,5462,2048,1536,1152,
 864,648,486,182,68,51,77,116,87,131,197,296,222,83,125,188,141,212,159,239,
 359,539,809,1214,455,683,1025,1538,288,216,162,30,11,17,26,2,0]
\end{codex}
Moreover, in this formulation, the conjecture implies that the
the elements of sequence generated by {\tt nsyr} are all different.

The next examples will show that computations for
{\tt nsyr} can be efficiently carried out
for giant numbers that, with the traditional bitstring representation,
would easily
overflow the memory of a computer with more transistors than
the atoms in the known universe.

And finally something we are quite sure has never been computed before,
we can also start with a {\em tower of exponents 100 levels tall}:
\begin{codex}
*GCat> take 100 (map(n . catsize) (nsyr (bestCase (t 100))))
[100,199,297,298,300,...,440,436,429,434,445,439]
\end{codex}
Note that we have only computed the decimal equivalents
of the representation complexity {\tt catsize} 
of these numbers, which, obviously,
would not fit themselves in a decimal representation.

A slightly longer computation (taking a few minutes)
can be also performed on a twin tower of exponents {\tt 101} 
and {\tt 103} levels tall like in
\begin{codex}
*GCat> take 2 (map(n.catsize) (nsyr 
         (add (bestCase (t 101)) (bestCase (t 103)))))
[10206,10500]
\end{codex}
where the Catalan representation size at {\tt 10500}, 
proportional to the product of the representation sizes of the operands, 
slows down computation but
still keeps it in a tractable range.

\section{Discussion} \label{disc}
Our Catalan families based numbering system provides compact
representations of giant numbers and can perform interesting
computations intractable with their bitstring-based counterparts.

This ability comes from the fact that
our canonical tree representation, in contrast to the traditional
binary representation supports constant average time and
space application of exponentials.

We have not performed a precise time and space complexity analysis
(except for the the constant average-time operations), but our
experiments indicate that (low) polynomial bounds are likely
for addition and subtraction with worst cases of size 
expansion happening with towers of exponents, where results 
are likely to be proportional to the product of the 
height of the towers, as illustrated in subsection \ref{collatz}.

Our multiplication and division algorithms are derived from
relatively simple traditional ones, with some focus on
taking advantage of large blocks of $0$ and $1$ digits.
However, it would be interesting to further explore
asymptotically better algorithms like Karatsuba multiplication
or division based on Newton's method.

As most numbers have high Kolmogorov complexity, it makes sense
to extend our type class mechanism to devise a {\em hybrid} numbering
system that switches between representations as needed, to
delegate cases where there are no benefits, to the underlying
bitstring representation. This is likely to be needed for designing 
a practical extension with Catalan representations
for the arithmetic subsystem
used in a programming language like Haskell that
benefits from the very fast C-based GMP library.

\section{Conclusion} \label{concl}

We have described through a type class mechanism an 
arithmetic system working on members of the  Catalan family of
combinatorial objects, that takes advantage of compact
representations of some giant numbers and can perform interesting
computations intractable with their bitstring-based counterparts.

This ability comes from the fact that tree representation, in contrast to the traditional
binary representation, supports constant average time and
space application of exponentials.

The resulting numbering system is {\em canonical} - each natural
number is represented as a unique object.
Besides unique decoding, canonical representations allow
testing for {\em syntactic equality}.
It is also {\em generic} -- no commitment is made to a particular member
of the Catalan family -- our type class provides all
the arithmetic operations to several instances, including typical
members of the Catalan family together with the usual natural numbers.

While these algorithms share similar complexity
with those described in \cite{ictac14tarau},
which focuses on instance {\tt M}
of our type class {\tt Cat}, 
the generic implementation in this paper
enables one to perform efficient arithmetic
operations with any of the 58 known
instances of the Catalan family 
described in \cite{stanleyAD,stanleyEC},
some of them with geometric, combinatorial, 
algebraic, formal languages, number and set theoretical
or physical flavor. 
Therefore, we believe that
this generalization is significant and it  opens the door
to new and possibly unexpected applications.

\section*{Acknowledgement} 
This research has been supported by NSF grant \verb~1423324~.


\bibliographystyle{splncs03}
\bibliography{../go/tarau,theory,proglang}

\section*{Appendix}

\subsection*{A subset of Haskell as an executable function notation}

We mention, for the benefit of the
reader unfamiliar with Haskell, that a notation like {\tt f x y} stands for $f(x,y)$,
{\tt [t]} represents sequences of type {\tt t} and a type declaration
like {\tt f :: s -> t -> u} stands for a function $f: s \times t \to u$
(modulo Haskell's ``currying'' operation, given the isomorphism between 
the function spaces ${s \times t} \to u$ and ${s \to t} \to u$). 

Our Haskell functions are always represented as sequences
of recursive equations guided by pattern matching, conditional
to constraints (boolean 
relations {\em following the ``\verb~|~'' symbol and before
the ``\verb~=~'' symbol}) in an equation.

Locally scoped helper functions are defined in Haskell
after the ``{\tt where}'' keyword, using the same equational style.
The composition of functions {\tt f} and {\tt g} is denoted {\tt f . g}.
It is  customary in Haskell
to write $f=g$ instead of $f~x=g~x$ (``point-free'' notation).
We  make some use of Haskell's ``call-by-need'' evaluation
that allows us to work with infinite
sequences, like the {\tt [0..]} infinite list notation, 
as well as higher order functions
(having other functions as arguments). Note also that the result
of the last evaluation is stored in the special Haskell
variable {\tt it}. 

By restricting ourselves to this {\em ``Haskell~- -''}
subset, our code can also be easily transliterated into
a system of rewriting rules, other pattern-based functional
languages as well as deterministic Horn Clauses. 

\end{document}